\documentclass[10pt,conference]{IEEEtran}

\ifx\pdfoutput\undefined
\usepackage[hypertex]{hyperref}
\else
\usepackage[pdftex,hypertexnames=false,hyperfigures=false,pdfview=FitH,pdfstartview=FitH]{hyperref}
\fi

\usepackage{amsmath,amssymb}
\usepackage{cite}

\makeatletter
\ifx\pdftexversion\@undefined 
  \usepackage{QC,colordvi}
\else
  \input{pdfcolor}
  \def\White#1{\pdfsetcolor{\cmykWhite}#1\pdfsetcolor{\cmykBlack}}

  \usepackage{QC}
\fi
\makeatother

\def\ol#1{\overline{#1}}
\def\C{{\mathbb C}}
\def\F{{\mathbb F}}
\def\entspricht{\mathrel{\hat{=}}}
\def\dist{\mathop{\rm dist}}
\def\wgt{\mathop{\rm wgt}}
\def\bra#1{\langle#1|}
\def\ket#1{|#1\rangle}

\long\def\dummy#1{}

\newtheorem{theorem}{Theorem}
\newtheorem{definition}[theorem]{Definition}
\newtheorem{lemma}[theorem]{Lemma}
\newtheorem{remark}[theorem]{Remark}

\begin{document}

\title{Non-Additive Quantum Codes from\\Goethals and Preparata Codes}

\author{\authorblockN{Markus Grassl}
\authorblockA{
Institute for Quantum Optics and Quantum Information\\
Austrian Academy of Sciences\\
Technikerstra{\ss}e 21a, 6020 Innsbruck, Austria\\
Email: markus.grassl@oeaw.ac.at}
\and
\authorblockN{Martin R\"otteler}
\authorblockA{
NEC Laboratories America, Inc.\\
4 Independence Way, Suite 200\\
Princeton, NJ 08540, USA\\
Email: mroetteler@nec-labs.com
}
}

\maketitle

\begin{abstract}
We extend the stabilizer formalism to a class of non-additive quantum
codes which are constructed from non-linear classical codes.  As an
example, we present infinite families of non-additive codes which
are derived from Goethals and Preparata codes.
\end{abstract}

\section{Introduction}
Recently, several new non-additive quantum error-correcting codes
(QECCs) have been constructed that have higher dimension than additive
QECCs with the same length and minimum
distance\cite{YCLO07,CSSZ07,YCO07}.  The first example of such a code
is the code $((5,6,2))$ of \cite{RHSS97} which has been found via
numerical optimization. Afterwards, the code has been identified as
the span of a particular state and its image under five unitary
transformations (see also \cite{GrBe97}).  The recently discovered
codes $((9,12,3))$ and $((10,24,3))$ (see \cite{YCLO07,YCO07}) start
with a so-called graph state which corresponds to a stabilizer state,
i.\,e., a stabilizer code with parameters $[[n,0,d]]$ (see
\cite{ScWe02,GKR02}).  A basis of the quantum code is obtained by this
initial state together with its image under some tensor products of
Pauli matrices and identity (Pauli operators).  The distance between
any pair of these states can be defined as the minimal weight of a
Pauli operator transforming one state into the other (see below).  The
problem of finding a code of high dimension can be stated as finding a
maximal clique in a \emph{search graph} whose vertices are all images
of the initial state.  There is an edge between two states if their
distance is at least the prescribed minimum distance.  Using the
formalism of graph states, Cross et al. show in \cite{CSSZ07} that it
is sufficient to consider $Z$-only operators in order to define the
search graph, but this has the disadvantage that the distance between
two of these states is not necessarily equal to the number of $Z$
operators in the tensor product.  Moreover, constructing a code
$((n,K,d))$ requires to find a clique of size $K$ in a search graph
with $2^n$ vertices.  Fixing one basis state, the graph can be
slightly simplified.  However, for the code $((10,24,3))$ the
simplified graph still has 678 vertices and 149.178 edges.

A different approach for constructing non-additive QECCs based on
Boolean functions and projection operators has been presented in
\cite{AgCa07}.  Evaluating the Boolean function at the projection
operators, the code is given as the sum of the image of products of
the projections. Finally, we mention the non-additive QECCs obtained
by the method given in \cite{RoVa99}. Those codes have the same
parameters as the additive CSS codes that can be obtained using the
same underlying binary codes.

In this paper, we extend the approach of combining unitary images of
stabilizer {\em states} to arbitrary stabilizer {\em codes} as
starting point.  For a stabilizer code $[[n,k,d]]$, the search graph
has only $2^{n-k}$ vertices, and a clique of size $K$ yields a quantum
code of dimension $K\times 2^k$.  What is more, we present some
infinite families of non-additive quantum codes which are constructed
using non-linear binary codes, avoiding the NP-hard problem of finding
a maximal clique in the search graph. Finally, we show how to obtain
encoding circuits for the resulting non-additive codes using encoding
circuits for the underlying stabilizer codes.

\section{A brief review of the stabilizer formalism}
We start with a brief review of the stabilizer formalism for quantum
error-correcting codes and the connection to additive codes over
$GF(4)$ (see, e.\,g., \cite{Got96,CRSS98}). A stabilizer code encoding
$k$ qubits into $n$ qubits having minimum distance $d$, denoted by
${\cal C}=[[n,k,d]]$, is a subspace of dimension $2^k$ of the complex
Hilbert space $(\C^2)^{\otimes n}$ of dimension $2^n$.  The code is
the joint eigenspace of a set of $n-k$ commuting operators
$S_1,\ldots,S_{n-k}$ which are tensor products of the Pauli matrices
\[
\sigma_x=\left(\begin{matrix}0&1\\1&0\end{matrix}\right),\quad
\sigma_y=\left(\begin{matrix}0&-i\\i&0\end{matrix}\right),\quad
\sigma_z=\left(\begin{matrix}1&0\\0&-1\end{matrix}\right),
\]
or identity.  The operators $S_i$ generate an Abelian group ${\cal S}$
with $2^{n-k}$ elements, called the \emph{stabilizer} of the code.  It
is a subgroup of the $n$-qubit Pauli group ${\cal P}_n$ which itself
is generated by the tensor product of $n$ Pauli matrices and identity.
We further require that ${\cal S}$ does not contain any non-trivial
multiple of identity.  The {\em normalizer} of ${\cal S}$ in ${\cal
  P}_n$, denoted by ${\cal N}$, acts on the code ${\cal C}=[[n,k,d]]$.
It is possible to identify $2k$ logical operators
$\ol{X}_1,\ldots,\ol{X}_k$ and $\ol{Z}_1,\ldots,\ol{Z}_k$ such that
these operators commute with any element in the stabilizer ${\cal S}$,
and such that together with ${\cal S}$ they generate the normalizer
${\cal N}$ of the code.  The operators $\ol{X}_i$ mutually commute,
and so do the operators $\ol{Z}_j$. The operator $\ol{X}_i$
anti-commutes with the operator $\ol{Z}_j$ if $i=j$ and otherwise
commutes with it.

It has been shown that the $n$-qubit Pauli group corresponds to a
symplectic geometry and that one can reduce the problem of
constructing stabilizer codes to finding additive codes over $GF(4)$
that are self-orthogonal with respect to a symplectic inner product
\cite{CRSS96,CRSS98}. Up to a scalar multiple, the elements of ${\cal
  P}_1$ can be expressed as $\sigma_x^a\sigma_z^b$ where
$(a,b)\in\F_2^2$ is a binary vector. Choosing the basis $\{1,\omega\}$
of $GF(4)$, where $\omega$ is a primitive element of $GF(4)$ with
$\omega^2+\omega+1=0$, we get the following correspondence between the
Pauli matrices, elements of $GF(4)$, and binary vectors of length two:
\[
\begin{array}{c|c|c}
\text{operator} & GF(4) & \F_2^2\\
\hline
I        & 0        & (00)\\
\sigma_x & 1        & (10)\\
\sigma_y & \omega^2 & (11)\\
\sigma_z & \omega   & (01)
\end{array}
\]
This mapping extends naturally to tensor products of $n$ Pauli matrices
being mapped to vectors of length $n$ over $GF(4)$ or binary vectors
of length $2n$.  We rearrange the latter in such a way that the first
$n$ coordinates correspond to the exponents of the operators
$\sigma_x$ and write the vector as $(a|b)$, i.\,e.,
\begin{equation}\label{eq:binary_rep}
g=\sigma_x^{a_1}\sigma_z^{b_1}\otimes\ldots\otimes\sigma_x^{a_n}\sigma_z^{b_n}
\entspricht(a|b)=(g^X|g^Z).
\end{equation}
Two operators corresponding to the binary vectors $(a|b)$ and $(c|d)$
commute if and only if the symplectic inner product
$a\cdot d-b\cdot c=0$. In terms of the binary representation, the
stabilizer corresponds to a binary code $C$ which is self-orthogonal
with respect to this symplectic inner product, and the normalizer
corresponds to the symplectic dual code $C^*$.  The stabilizer
together with the logical operators $\ol{Z}_i$ generate a self-dual
code.  In terms of the correspondence to vectors over $GF(4)$, the
stabilizer and normalizer correspond to an additive code over
$GF(4)$ and its dual with respect to an symplectic inner product,
respectively, which we will also denote by $C$ and $C^*$.  The
minimum distance $d$ of the quantum code is given as the minimum
weight in the set $C^*\setminus C$ which is lower bounded by the
minimum distance $d^*$ of $C^*$.  If $d=d^*$, the code is said to be
\emph{pure}, and for $d\ge d^*$, the code is said to be \emph{pure
  up to $d^*$}.

\section{The union of stabilizer codes}
Note that we have defined a stabilizer code ${\cal C}$ as the joint
eigenspace of the commuting operators $S_i$ generating the stabilizer
${\cal S}$.  The term \emph{stabilizer} suggests that the code is the
joint $+1$ eigenspace of the operators.  However, for each of the
generators $S_i$ we may choose either the eigenspace with eigenvalue
$+1$ or the eigenspace with eigenvalue $-1$.  This gives rise to a
decomposition of the space $(\C^2)^{\otimes n}$ into $2^{n-k}$
mutually orthogonal spaces which can be labeled by the eigenvalues of
the $n-k$ generators $S_i$, or equivalently, by the characters $\chi$
of the stabilizer group ${\cal S}$.  Moreover, the $n$-qubit Pauli
group ${\cal P}_n$ acts transitively on these spaces.

From now on we fix the code ${\cal C}$ as the the joint $+1$
eigenspace corresponding to the trivial character.  Let $t\in{\cal
  P}_n$ be an arbitrary $n$-qubit Pauli operator.  Then we can define
a character of ${\cal S}$ on the generators $S_i$ as follows
\[
\chi_{t}(S_i):=\begin{cases}
+1 & \text{if $t$ and $S_i$ commute,}\\
-1 & \text{if $t$ and $S_i$ anti-commute.}
\end{cases}
\]
As the elements of the normalizer ${\cal N}$ commute with all elements
of the stabilizer ${\cal S}$, two elements $t_1$ and $t_2$ define the
same character if $t_2^{-1}t_1\in{\cal N}$.  Hence the set of
characters corresponds to the cosets of ${\cal N}$ in ${\cal P}_n$.
If ${\cal T}$ is a set of coset representatives, we can write the
decomposition of the full space as
\begin{equation}\label{eq:full_decomposition}
(\C^2)^{\otimes n}=\bigoplus_{t\in{\cal T}} t {\cal C}.
\end{equation}
Note that measuring the eigenvalues of the generators $S_i$ projects
onto one of these space $t{\cal C}$, corresponding to the character
$\chi_t$ given by the sequence of eigenvalues.  In terms of the
classical codes, the eigenvalues correspond to an error-syndrome which
is obtained by computing the symplectic inner product of the received
vector with the $n-k$ vectors corresponding to the generators of the
stabilizer, i.\,e., a basis of the code $C$.  For all vectors of the
dual code $C^*$ corresponding to the normalizer ${\cal N}$, the inner
product is zero. So the different spaces $t{\cal C}$ correspond to
cosets $C^*+t$ of the code $C^*$.

As for a fixed code ${\cal C}$ two spaces $t_1{\cal C}$ and $t_2{\cal
  C}$ are either identical or orthogonal, we can define the distance
of them as follows:
\begin{equation}\label{def:distance_pauli}
\dist(t_1{\cal C},t_2{\cal C}):=\min\{\wgt(p):p\in{\cal P}_n\mid p t_1{\cal C}=t_2{\cal C}\}.
\end{equation}
Here $\wgt(p)$ is the number of tensor factors in the $n$-qubit Pauli
operator $p$ that are different from identity.  Clearly,
$\dist(t_1{\cal C},t_2{\cal C})=\dist(t_2^{-1}t_1{\cal C},{\cal C})$.
The distance (\ref{def:distance_pauli}) can also be expressed in terms
of the associated vectors over $GF(4)$.
\begin{lemma}\label{lemma:distance}
The distance of the spaces $t_1{\cal C}$ and $t_2{\cal C}$ equals the
minimum weight in the coset $C^*+t_1-t_2$, where we use $t_i$ to
denote both an $n$-qubit Pauli operator and the corresponding vector over $GF(4)$.
\end{lemma}
\begin{proof}
Direct computation shows\\[-3.5ex]
\begin{alignat*}{3}
\dist(t_1{\cal C},t_2{\cal C})
&=\dist(C^*+t_1,C^*+t_2)\\
&=\dist(C^*+(t_1-t_2),C^*)\\
&=\min\{\wgt(c+t_1-t_2):c\in C^*\}\\
&=\min\{\wgt(v):v \in C^*+t_1-t_2\}.
\end{alignat*}
\vskip-1.25\baselineskip
\end{proof}

With this preparation, we are ready to present the general
construction of the union of stabilizer codes (see also
\cite{GrBe97}). The quantum code will be defined as the span of some
of the summands in (\ref{eq:full_decomposition}).
\begin{definition}[union stabilizer code]\label{def:unioncode}
Let ${\cal C}_0=[[n,k,d_0]]$ be a stabilizer code and let ${\cal
  T}_0=\{t_1,\ldots,t_K\}$ be a subset of the coset representatives of
the normalizer ${\cal N}_0$ of the code ${\cal C}_0$ in ${\cal P}_n$.
Then the \emph{union stabilizer code} is defined as
\[
{\cal C}=\bigoplus_{t\in {\cal T}_0} t{\cal C}_0.
\]
Without loss of generality we assume that ${\cal T}_0$ contains
identity. With the union stabilizer code ${\cal C}$ we associate the
(in general non-additive) \emph{union normalizer code} given by
\[
C^*=\bigcup_{t\in {\cal T}_0} C_0^*+t=\{c+t_i:c \in C_0^*, i=1,\ldots,K\},
\]
where $C_0^*$ denotes the additive code associated with the normalizer
${\cal N}_0$ of the stabilizer code ${\cal C}_0$.  We will refer to
both, the vectors $t_i$ and the corresponding unitary operators, as
\emph{translations}.
\end{definition}
A union stabilizer code can be defined in terms of binary vectors as
shown in Fig. \ref{fig:union_generators}.  The first $n-k$ rows
correspond to the binary vectors (cf. (\ref{eq:binary_rep}))
associated with the generators $S_i$ of the stabilizer ${\cal S}$ of
the code ${\cal C}_0$.  They generate the classical code $C_0$.  The
next $k$ rows correspond to the logical operators $\ol{Z}_j$, followed
by the $k$ logical operators $\ol{X}_i$.  The last $K$ rows correspond
to the $K$ translations $t_i$ defining the cosets of the classical
code $C_0^*$ and the unitary images of the stabilizer code ${\cal
  C}_0$, respectively.  We use curly brackets to stress the fact that
the set of operators ${\cal T}_0$ need not be closed under group
operation.  In general, the quantum code is not invariant under these
\emph{generalized logical $X$-operators}.  On the other hand, if
${\cal T}_0$ is closed under group operation, the resulting code will
be a stabilizer code where a basis of ${\cal T}_0$ defines an
additional set of logical $X$-operators.

\begin{figure}
\[
\def\dashes{\rule{2mm}{0.5pt}\rule{1mm}{0pt}\rule{2mm}{0.5pt}\rule{1mm}{0pt}\rule{2mm}{0.5pt}\rule{1mm}{0pt}\rule{2mm}{0.5pt}\rule{1mm}{0pt}\rule{2mm}{0.5pt}}
\begin{array}{c}
\left(\def\arraystretch{1.3}
\begin{array}{c|c}
S^X_1&S^Z_1\\
\vdots&\vdots\\
S^X_{n-k}&S^Z_{n-k}\\[0.75ex]
\hline
\ol{Z}^X_1&\ol{Z}^Z_1\rule{0pt}{2.75ex}\\
\vdots&\vdots\\
\ol{Z}^X_k&\ol{Z}^Z_k\\[-1ex]
\multicolumn{1}{@{}c|}{\dashes}&
\multicolumn{1}{c@{}}{\dashes}\\
\ol{X}^X_1&\ol{X}^Z_1\\
\vdots&\vdots\\
\ol{X}^X_k&\ol{X}^Z_k\\
\hline
\end{array}
\right)\\
\left\{\def\arraystretch{1.3}
\begin{array}{c|c}
\hbox to 12mm{\hfill$t^X_1$\hfill}&
\hbox to 12mm{\hfill$t^Z_1$\hfill}\\
\vdots&\vdots\\
t^X_K&t^Z_K
\end{array}
\right\}
\end{array}
\]
\caption{Arrangements of the vectors associated with a union
  stabilizer code.\label{fig:union_generators}} 
\end{figure}

\begin{theorem}
Let ${\cal C}$ be a union stabilizer code as in Definition
\ref{def:unioncode}.  The dimension of ${\cal C}$ is $|{\cal
  T}_0|2^k=K2^k$, and the minimum distance is lower bounded by the
minimum distance $d$ of the union normalizer code $C^*$.
\end{theorem}
\begin{proof}
As ${\cal T}_0$ is a subset of the coset representatives of the
normalizer ${\cal N}_0$, the spaces $t_i{\cal C}_0$, each of which has
dimension $2^k$, are mutually orthogonal. Hence the dimension of the
union code is $K2^k$.  Fixing an orthonormal basis
$\{\ket{c_j}:j=1,\ldots,2^k\}$ of the stabilizer code ${\cal C}_0$,
the set $\{t_i\ket{c_j}:i=1,\ldots,K,\; j=1,\ldots,2^k\}$ is
an orthonormal basis of the union stabilizer code.  Let $E\in{\cal
  P}_n$ be an $n$-qubit Pauli error of weight $0<\wgt(E)< d$.  For
basis states $\ket{c_{i,j}}=t_i\ket{c_j}\in t_i{\cal C}_0$ and
$\ket{c_{i',j'}}=t_{i'}\ket{c_{j'}}\in t_{i'}{\cal C}_0$ we consider
the inner product
\begin{equation}\label{eq:code_condition}
\bra{c_{i,j}} E\ket{c_{i',j'}}=
\bra{c_j}t_i^\dagger E t_{i'}\ket{c_{j'}}.
\end{equation}
For $i\ne i'$, we have $\dist(t_i{\cal C}_0,t_{i'}{\cal
  C}_0)=\min\{\wgt(c+t_i-t_{i'}):c\in C_0^*\}\ge d$, and
hence (\ref{eq:code_condition}) vanishes for $\wgt(E)<d$. For $i=i'$,
we get
\begin{equation}\label{eq:code_condition2}
\bra{c_j}t_i^\dagger E t_i\ket{c_{j'}}
=\pm\bra{c_j} E\ket{c_{j'}}.
\end{equation}
As the code ${\cal C}_0$ is pure up to the minimum distance of
$C^*_0\subset C^*$, equation (\ref{eq:code_condition2}) vanishes
as well.
\end{proof}
\begin{remark}
We note that similar to stabilizer codes, the true minimum distance of
a union stabilizer code might be higher.  The true minimum distance is
given by
\begin{alignat*}{2}
\min\{\dist(c+t_i,c'+t_{i'})&{}:t_i,t_{i'}\in{\cal T}_0,\\
&c,c'\in C^*_0|c+t_i-(c'+t_{i'})\notin \widetilde{C}_0\}\\
\hbox to 0pt{\hss$=\min\{\wgt(v):v \in (C^*-C^*)\setminus \widetilde{C}_0\}$,\hss}
\end{alignat*}
where $(C^*-C^*)=\{a-b:a,b \in C^*\}$ denotes the set of all
differences of vectors in $C^*_0$, and $\widetilde{C}_0$ is the
symplectic dual of the additive closure of $C^*$.
\end{remark}
In order to construct union quantum codes, we may start with a
stabilizer code ${\cal C}_0$ and use a \emph{search graph} whose
vertices are the mutually orthogonal translates $\{t{\cal C}_0:t \in
{\cal T}_0\}$ of the stabilizer code.  Two vertices are connected by
an edge if and only if the distance between them is at least $d$,
where $d$ is the desired minimum distance.  For simplicity, we also
require that the code ${\cal C}_0$ is pure up to $d_0\ge d$.  The
distance between two translates can be computed using Lemma
\ref{lemma:distance}.  This allows to use stabilizer codes ${\cal
  C}_0$ of arbitrary dimension, and hence allows to go beyond the case
of stabilizer states (or graph states) as, e.\,g., in
\cite{CSSZ07,YCO07}.  We note that the construction of non-additive
quantum codes of \cite{AKP04} is also based on taking the union of
orthogonal images of a stabilizer code.

\section{Union Stabilizer Codes from Binary Codes}
\subsection{CSS-like codes}
Given two linear binary codes $C_1=[n,k_1,d_1]$ and $C_2=[n,k_2,d_2]$
with $C_2^\bot\subset C_1$, the so-called CSS construction (see,
e.\,g., \cite{Stea99b}) yields a quantum error-correcting code ${\cal
  C}=[[n,k_1+k_2-n,d]]$ with $d\ge\min(d_1,d_2)$.  Starting with this
CSS code, we consider unions of cosets of the binary codes $C_i$,
i.\,e.,
\[
\widetilde{C}_i=\bigcup_{t^{(i)}\in{\cal T}_i} C_i+t^{(i)}
\]
such that the minimum distance of the codes $\widetilde{C}_i$ is at
least $\widetilde{d}\le d$.  Using the translations
$\{(t^{(1)}|t^{(2)}):t^{(1)}\in {\cal T}_1, t^{(2)}\in{\cal T}_2\}$ we
obtain a CSS-like union stabilizer code of dimension $|{\cal
  T}_1|\cdot|{\cal T}_2|\cdot 2^{k_1+k_2-n}$ whose minimum distance is
at least $\widetilde{d}$.  If $G_1=\left(\begin{smallmatrix}H_2\\\hline
    G_{12}\end{smallmatrix}\right)$ and
$G_2=\left(\begin{smallmatrix}H_1\\\hline
    G_{21}\end{smallmatrix}\right)$ are generator matrices of the
codes $C_i$, where $H_i$ is a generator matrix of the dual code
$C_i^\bot$, the corresponding vectors are as shown in Fig.
\ref{fig:union_CSS_generators}.

\begin{figure}[hbt]
\[
\def\dashes{\rule{2mm}{0.5pt}\rule{1mm}{0pt}\rule{2mm}{0.5pt}\rule{1mm}{0pt}\rule{2mm}{0.5pt}\rule{1mm}{0pt}\rule{2mm}{0.5pt}\rule{1mm}{0pt}\rule{2mm}{0.5pt}}
\begin{array}{c}
\left(\def\arraystretch{1.3}
\begin{array}{c|c}
H_2&0\\
0&H_1\\[0.75ex]
\hline
G_{12}&0\rule{0pt}{2.75ex}\\[-1ex]
\multicolumn{1}{@{}c|}{\dashes}&
\multicolumn{1}{c@{}}{\dashes}\\
0&G_{21}\\
\hline
\end{array}
\right)\\
\left\{\def\arraystretch{1.3}
\begin{array}{c|c}
\hbox to 12mm{\hfill$t_1^{(1)}$\hfill}&
\hbox to 12mm{\hfill$t_1^{(2)}$\hfill}\\[-1ex]
\vdots&\vdots\\[-0.8ex]
t_1^{(1)}&t_{K_2}^{(2)}\\
\vdots&\vdots\\
t_{K_1}^{(1)}&t_1^{(2)}\\[-1ex]
\vdots&\vdots\\[-0.8ex]
t_{K_1}^{(1)}&t_{K_2}^{(2)}
\end{array}
\right\}
\end{array}
\]
\caption{Arrangements of the vectors associated with a CSS-like union
  stabilizer code.\label{fig:union_CSS_generators}} 
\end{figure}

\subsection{Enlargement construction}
Steane has presented a construction that allows to increase the
dimension of a CSS code \cite{Stea99b}.  For this, he starts with the
CSS construction applied to a binary code $C=[n,k,d]$ which contains
its dual, yielding a CSS code ${\cal C}_0=[[n,2k-n,d]]$.  Using a code
$C'=[n,k'>k+1,d']$ which contains $C$, he obtains a quantum code
$[[n,k+k'-n,\min(d,\lceil 3d'/2\rceil)]]$.  The resulting code can
also be considered as a union stabilizer code.  If $D$ is a generator
matrix of the complement of $C$ in $C'$ and $A$ is a fixed point free,
invertible linear map, the translations can be defined as
\[
{\cal T}_0=\{(vD|vAD):v\in \F_2^{k'-k}\}.
\]
The key observation \cite{Stea99b} for proving the lower bound on the
minimum distance is that the weight of an operator $g=(g^X|g^Z)$ can
be expressed in terms of the Hamming weight of the binary vectors and
their sum:
\[
\wgt((g^X|g^Z))=\frac{1}{2}(\wgt(g^X)+\wgt(g^Z)+\wgt(g^X+g^Z)).
\]
As ${\cal T}_0$ is closed under addition and the properties of $A$
ensure that $0\ne t^X\ne t^Z\ne 0$ for any non-zero element
$(t^X|t^Z)\in{\cal T}_0$, the weight of all three binary vectors is
lower bounded by $d'$.

\section{Quantum Codes from Reed-Muller, Goethals, and Preparata Codes}
Using the CSS-like construction of the previous section, we now
construct some families of non-additive quantum codes.  For this, we
use the Goethals codes ${\cal G}(m)$ and the Preparata codes which are
nonlinear binary codes of length $n=2^m$ for $m\ge 4$, $m$ even. Some
of the properties of these codes are summarized as follows (see,
e.\,g., \cite{Her90}):
\begin{itemize}
\item Both the Goethals code ${\cal G}(m)$ and the Preparata code
${\cal P}(m)$ are unions of cosets of the Reed-Muller code ${\cal
  R}(m):={\rm RM}(m-3,m)$.  Furthermore, they are nested subcodes of 
${\rm RM}(m-2,m)$, i.\,e.,
\[
{\rm RM}(m-3,m)\subset {\cal G}(m)\subset {\cal P}(m)\subset {\rm RM}(m-2,m).
\]
\item The parameters of the codes are
\begin{alignat*}{2}
  {\rm RM}(m-3,m)=
{\cal R}(m)&=[2^m,2^m-\tbinom{m}{2}-m-1,8]\\
{\cal G}(m)&=(2^m,2^{2^m-3m+1},8)\\
{\cal P}(m)&=(2^m,2^{2^m-2m},6)\\
{\rm RM}(m-2,m)&=[2^m,2^m-m-1,4].
\end{alignat*}
\end{itemize}

Steane has constructed a family of additive quantum codes from
Reed-Muller codes \cite{Stea99a}.  The codes are obtained applying the
enlargement construction of \cite{Stea99b} to the chain of codes
\begin{alignat*}{2}
{\rm RM}(r,m)\subset {\rm RM}(r,m)^\bot&={\rm RM}(m-r-1,m)\\
                                      &\subset{\rm RM}(m-r,m).
\end{alignat*}
In particular, for $r=2$ and $m\ge 5$ this yields additive QECCs
${\cal C}_1=[[2^m,2^m-\binom{m}{2}-2m-2,6]]$, while using only the CSS
construction, one obtains ${\cal C}_0=[[2^m,2^m-2\binom{m}{2}-2m-2,8]]$.
As the Goethals code ${\cal G}(m)$ is the union of $K_G=2^{\binom{m}{2}-2m+2}$
cosets of ${\cal R}(m)$, we can construct a CSS-like union
stabilizer code based on ${\cal C}_0$.  The minimum distance of the
resulting non-additive code is $8$ and its dimension is
$K_G^2\dim({\cal C}_0)=2^{2^m-6m+2}$. 

Replacing the Goethals code by the Preparata code ${\cal P}(m)$, we
have $K_P=2^{\binom{m}{2}-m+1}$ cosets of ${\cal R}(m)$. This results
in a CSS-like union stabilizer code with minimum distance $6$ and
dimension $K_P^2\dim({\cal C}_0)=2^{2^m-4m}$.

Both the union stabilizer codes based on Goethals codes and those
based on Preparata codes are superior to the additive codes derived
from Reed-Muller codes.  The parameters of the first codes in these
families are as follows:
\[\def\arraystretch{1.3}
\begin{array}{c|c|c}
\text{enlarged RM} & \text{Goethals}& \text{Preparata}\\
\hline
[[   64,  35, 6]] & ((  64, 2^{ 30}, 8)) & ((  64, 2^{ 40}, 6))\\{}
[[  256, 210, 6]] & (( 256, 2^{210}, 8)) & (( 256, 2^{224}, 6))\\{}
[[ 1024, 957, 6]] & ((1024, 2^{966}, 8)) & ((1024, 2^{984}, 6))
\end{array}
\]
However, applying the enlargement construction to extended primitive
BCH codes results in stabilizer codes with parameters
$[[2^m,2^m-5m-2,8]]$ and $[[2^m,2^m-3m-2,6]]$ (see \cite{Stea99b}).

\begin{figure}
\[
\def\dashes{\rule{2mm}{0.5pt}\rule{1mm}{0pt}\rule{2mm}{0.5pt}\rule{1mm}{0pt}\rule{2mm}{0.5pt}\rule{1mm}{0pt}\rule{2mm}{0.5pt}\rule{1mm}{0pt}\rule{2mm}{0.5pt}}
\def\sdashes{\rule{2mm}{0.5pt}\rule{1mm}{0pt}\rule{2mm}{0.5pt}\rule{1mm}{0pt}\rule{2mm}{0.5pt}}
\begin{array}{@{}c@{}}
\left(
\begin{array}{@{}c@{}}
\begin{array}{@{}c|c@{}}
X&Z\\
\end{array}\\
\hline
\ol{Z}\rule{0pt}{2.75ex}\\[-1ex]
\dashes\\
\ol{X}\\
\hline
\end{array}
\right)\\[4.25ex]
\left\{
\hbox to 14mm{\rule[-1ex]{0pt}{2.2ex}\hfill ${\cal T}_0$\hfill}
\right\}
\end{array}
\stackrel{Q_1}{\longrightarrow}
\begin{array}{@{}c@{}}
\left(
\begin{array}{@{}c@{\,}|@{\,}c@{}}
00&I0\\
\hline
00&0I\rule{0pt}{2.75ex}\\[-1ex]
\sdashes&\sdashes\\
0I&00\\
\hline
\end{array}
\right)\\[4.25ex]
\left\{
\hbox to 17mm{\rule[-1ex]{0pt}{2.2ex}\hfill ${\cal T}_0^{Q_1}$\hfill}
\right\}
\end{array}
\stackrel{Q_c}{\longrightarrow}
\begin{array}{@{}c@{}}
\left(
\begin{array}{@{}c@{\,}|@{\,}c@{}}
\hbox to 8mm{\hfill$00$\hfill}&\hbox to 8mm{\hfill$I0$\hfill}\\
\hline
00&0I\rule{0pt}{2.75ex}\\[-1ex]
\sdashes&\sdashes\\
0I&00\\
\hline
\end{array}
\right)\\[4.25ex]
\left\{
\begin{array}{@{}c|c@{}}
\hbox to 7mm{\hfill$c_10$\hfill}&\hbox to 7mm{\hfill$00$\hfill}\\[-1ex]
\vdots&\vdots\\
c_K0&00
\end{array}
\right\}\\[-4.5ex]
\end{array}
\rule[-10ex]{0pt}{0pt}
\]
\caption{Transformation of the union stabilizer code given by the
inverse encoding circuits $Q_1$ and $Q_c$.\label{fig:encoding_transform}}
\end{figure}

\begin{figure*}
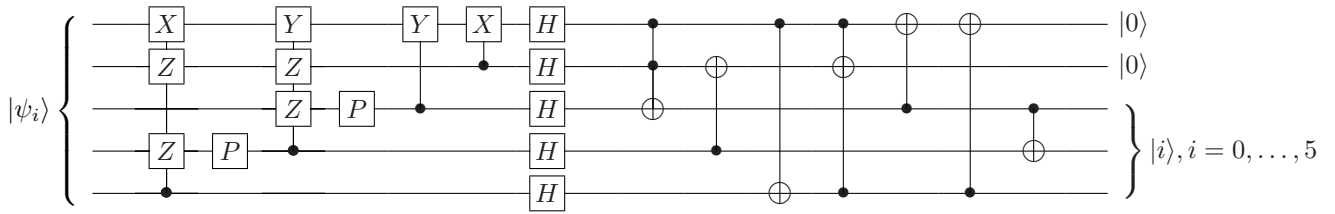

\centerline{
\unitlength0.8\unitlength
\inputwires[,,$\ket{\psi_i}\left\{\rule{0pt}{50\unitlength}\right.$](5)
\CGate(5,1,5){$X$}\kern-30\unitlength%
\CGate(4,1,4){$Z$}\kern-30\unitlength%
\CGate(3,2,3){$Z$}
\OneQubitGate(4,5){$P$}
\CGate(4,1,5){$Y$}\kern-30\unitlength%
\CGate(3,1,4){$Z$}\kern-30\unitlength%
\CGate(2,1,3){$Z$}
\OneQubitGate(3,5){$P$}
\CGate(3,1,5){$Y$}
\CGate(2,1,5){$X$}
\rlap{\OneQubitGate[4](1,1){$H$}}%
\rlap{\OneQubitGate[3](1,1){$H$}}%
\rlap{\OneQubitGate[2](1,1){$H$}}%
\rlap{\OneQubitGate[1](1,1){$H$}}%
\OneQubitGate(1,1){$H$}
\wires[20](5)
\multCNOT[1,2](3,5)
\CNOT(4,2,5) 
\CNOT(1,5,5) 
\multCNOT[1,5](2,5) 
\CNOT(3,1,5) 
\CNOT(5,1,5) 
\CNOT(3,4,5)
\outputwires[$\ket{0}$,$\ket{0}$,,{$\left.\rule{0pt}{25\unitlength}\right\}\ket{i},
  i=0,\ldots,5$}](5)
\rule{50\unitlength}{0pt}
}
\caption{Inverse encoding circuit for the non-additive code
  $((5,6,2))$. The first set of gates including the Hadamard
  transformations implements the inverse encoding circuit $Q_1$ for
  the stabilizer code $[[5,0,3]]$, followed by 5 ${\rm CNOT}$ and 2
  Toffoli gates implementing the classical circuit $Q_c$. 
\label{fig:encoding_circuit}}
\end{figure*}

\section{Encoding Circuits}
In \cite{GRB03} we have shown how to compute a quantum circuit
consisting of Clifford gates only that transforms any stabilizer
${\cal S}$ given by the binary $(n-k)\times 2n$ matrix $(X|Z)$ into
the stabilizer of a trivial code given by $(0|I0)$, where $I$ is an
identity matrix of size $n-k$.  The corresponding trivial code
corresponds to the mapping $\ket{\phi}\mapsto\ket{0\ldots
  0}\ket{\phi}$.  We denote the resulting quantum circuit that
corresponds to the inverse encoding circuit of $(X|Z)$ by $Q_1$. Note
that we can apply $Q_1$ to all the operators defining the code as
illustrated in Fig.  \ref{fig:encoding_transform}.  Further note, that
for the trivial stabilizer code, the ``encoded'' $X$- and
$Z$-operators are weight-one Pauli operators $\sigma_x$ and
$\sigma_z$, respectively, acting on the last $k$ qubits. As the
transformed translations ${\cal T}_0^{Q_1}$ define cosets of the
normalizer, we may choose them such that they are tensor products of
operators $\sigma_x$ and identity acting on the first $n-k$ qubits
only.  Then we have the trivial union code spanned by a set of $K2^k$
basis states of the form $\ket{c_i}\ket{j}$, where
$\{\ket{j}:j=0,\ldots,2^k-1\}$ is the computational basis of $k$
qubits and $\{c_i:i=0,\ldots,K-1\}$ is a set of bit strings of length
$n-k$.  In order to obtain a standard basis for our input space of
dimension $K2^k$, we need a quantum circuit $Q_c$ mapping
$\ket{c_i}\mapsto\ket{i}$ for $i=0,\ldots,K-1$.  Note that this is a
purely classical circuit which can be realized, e.\,g., using
$\sigma_x$, ${\rm CNOT}$ gates, and Toffoli gates.

We illustrate this for the non-additive code $((5,6,2))$ which is a
union stabilizer code derived from the stabilizer state ${\cal
  C}_0=[[5,0,3]]$.  For a stabilizer code with $k=0$, there are no
encoded $X$- and $Z$-operators.  So the code $((5,6,2))$ is specified
by five generators of the stabilizer and six translations.  Using an
inverse encoding circuit $Q_1$, these operators are transformed as
follows:
\begin{equation}\label{eq:5bit_encoding}
\begin{array}{c}
\left(
\begin{array}{@{}ccccc@{}}
X & X & X & X & X\\ 
X & X & Z & I & Z\\ 
X & Z & I & Z & X\\ 
Y & I & Y & Z & Z\\ 
Y & Z & Z & Y & I\\
\hline
\end{array}
\right)\\
\left\{
\begin{array}{ccccc}
I & I & I & I & I\\ 
I & I & Z & Z & X\\ 
I & I & I & X & X\\ 
I & I & I & Z & Y\\ 
I & I & Z & Y & Y\\ 
I & I & Z & X & Z
\end{array}
\right\}
\end{array}
\stackrel{Q_1}{\longrightarrow}
\begin{array}{c}
\left(
\begin{array}{@{}c|c@{}}
00000 & 10000 \\ 
00000 & 01000 \\ 
00000 & 00100 \\ 
00000 & 00010 \\ 
00000 & 00001 \\
\hline
\end{array}
\right)\\
\left\{
\begin{array}{@{}c|c@{}}
00000 & 00000 \\ 
01010 & 00000 \\ 
11011 & 00000 \\ 
01111 & 00000 \\ 
11100 & 00000 \\
10010 & 00000
\end{array}
\right\}
\end{array}
\end{equation}
It remains to find a (classical) quantum circuit $Q_c$ that maps the
six binary strings on the right hand side of (\ref{eq:5bit_encoding})
to say the binary representations of $i=0,\ldots,5$.  Using a
breath-first search among all circuits composed of $\sigma_x$, ${\rm
  CNOT}$, and Toffoli gates, we found the minimal realization shown
together with the circuit $Q_1$ in Fig. \ref{fig:encoding_circuit}.

\section{Conclusions}
The approach presented in this paper generalizes naturally to
the construction of non-additive quantum codes for higher dimensional
systems.  In order to obtain other families of non-additive quantum
codes, it is interesting to study classical non-linear codes which can
be decomposed into cosets of linear codes, similar to the Preparata and
Goethals codes.

\section*{Acknowledgments}
We acknowledge fruitful discussions with Vaneet Aggarwal and Robert
Calderbank.  Markus Grassl would like to thank NEC Labs., Princeton
for the hospitality during his visit.  This work was partially
supported by the FWF (project P17838).

\end{document}